\documentclass[authoryear,3p]{elsarticle}
\usepackage[utf8]{inputenc}
\usepackage{amsmath}
\usepackage{amsfonts}
\usepackage{amssymb}
\usepackage{amsfonts}
\usepackage{amsthm}
\usepackage{amscd}
\usepackage{color}
\usepackage{graphics}
\usepackage{graphicx}
\usepackage{psfrag}
\usepackage[toc,page]{appendix}

\usepackage[utf8]{inputenc}
\usepackage{tikz}
\definecolor{ccccccc}{RGB}{204,204,204}
\definecolor{cffffff}{RGB}{255,255,255}
\definecolor{cff0000}{RGB}{255,0,0}
\definecolor{c0000ff}{RGB}{0,0,255}
\definecolor{c00ff00}{RGB}{0,255,0}

\usepackage[all]{xy}

\usepackage{algorithm,algorithmic}

\usepackage{natbib}

\newtheorem{theorem}{\bf Theorem}[section]
\newtheorem{lemma}[theorem]{\bf Lemma}
\newtheorem{proposition}[theorem]{\bf Proposition}

\newtheorem{remark}[theorem]{\bf Remark}
\newtheorem{definition}[theorem]{\bf Definition}
\newtheorem{example}[theorem]{\bf Example}
\newtheorem{corollary}[theorem]{\bf Corollary}

\newcommand{\av}[1]{\ensuremath{\mathcal{#1}}}

\newcommand{\vek}[1][h]{\ensuremath{\mathbf{#1}}}
\newcommand{\f}[1]{\mathbf{#1}}

\newcommand{\oo}[2][d]{{\cal O}_{#1}({\cal #2})}

\newcommand{\R}{\mathbb{R}}

\journal{Computer Aided Geometric Design}

\begin{document}

\sloppy

\begin{frontmatter}

\title{Smooth surface interpolation using patches with rational offsets}

\author[plzen1,plzen2]{Miroslav L\'{a}vi\v{c}ka\corref{cor1}}
\cortext[cor1]{Corresponding author}
\ead{lavicka@kma.zcu.cz}

\author[praha]{Zbyn\v{e}k \v{S}\'{i}r}
\ead{zbynek.sir@mff.cuni.cz}

\author[plzen1,plzen2]{Jan Vr\v{s}ek}
\ead{vrsekjan@kma.zcu.cz}

\address[plzen1]{Department of Mathematics, Faculty of Applied Sciences, University of West Bohemia,
         Univerzitn\'i~8,~306~14~Plze\v{n},~Czech~Republic}

\address[plzen2]{NTIS -- New Technologies for the Information Society, Faculty of Applied Sciences, University of West Bohemia, Univerzitn\'i 8, 306 14 Plze\v{n}, Czech~Republic}

\address[praha]{Mathematical Institute,  Charles University in Prague, Sokolovská 83, 186 75 Praha 8, Czech~Republic}

\begin{abstract}
We present a simple functional method for the interpolation of given data points and associated normals with surface parametric patches with rational normal fields. We give some arguments why a dual approach is especially convenient for these surfaces, which are traditionally called Pythagorean normal vector (PN) surfaces. Our construction is based on the isotropic model of the dual space to which the original data are pushed. Then the bicubic Coons patches are constructed in the isotropic space and then pulled back to the standard three dimensional space. As a result we obtain the patch construction which is completely local and produces surfaces with the global $G^1$~continuity.
\end{abstract}


\end{frontmatter}

\section{Introduction}\label{sec intro}

This paper is devoted to the Hermite interpolation with the surfaces possessing Pythagorean normal vector fields (PN surfaces). These surfaces were introduced by \cite{Po95}. We can understand them as surface counterparts to the Pythagorean hodograph (PH) curves first studied by \cite{FaSa90}. PN surfaces have rational offsets and thus provide an elegant solution to many offset-based problems occurring in various practical applications. In particular, in the context of the computer-aided manufacturing, the tool path does not have to be approximated and it can be described exactly in the NURBS form, which is nowadays a standard format of the CAD/CAM applications.

For the survey of the theory and applications of PH/PN objects, see \cite{Fa08} and references therein. Many interesting theoretical questions related to this subject have been studied in the past years. Let us mention in particular the analysis of the geometric and algebraic properties of the offsets, such as the determination of the number and type of their components and the construction of their suitable rational parameterizations \citep{ArSeSe97,ArSeSe99,Ma99,SeSe00,VrLa10}.

Despite natural similarities between the PH curves and the PN surfaces, the two classes of Pythagorean objects exhibit also some important differences. For example, the set of all polynomial PH curves within the set of all rational PH/PN curves was exactly identified in \citep{FaPo96}. On the other hand for the PN surfaces only the rational ones are described explicitly using a dual representation and the subset of the polynomial ones have not been revealed yet. Polynomial solution of the Pythagorean condition in the surface case started in \citep{LaVr11} for cubic parameterizations and recently an approach based on bivariate polynomials with quaternion coefficients was presented by \cite{KoKrVi16}.
A survey discussing rational surfaces with rational offsets and their modelling applications can be found in \citep{KrPe10}.

The previous problem is also strongly related to the construction techniques for PN surfaces, in particular to the Hermite interpolation, which is the main topic of this paper. There exist many Hermite interpolation results for the polynomial and rational PH curves yielding piecewise curves of various continuity, see \citep{Fa08,KoLa14}.
Concerning direct algorithms for the interpolations with PN surfaces the situation is different. By `direct' we mean in this context the construction of the object together with its PN parameterization. Indeed, some constructions for special surfaces, which become PN only after a suitable reparameterization, were designed.  For instance in \cite{BaJuKoLa08}, there was designed a method for the construction of the exact offsets of quadratic triangular B\'{e}zier surface patches, which are in fact PN surfaces. However their PN parameterizations were gained via certain reparameterization. A similar approach based on reparameterization was also used in the paper \citep{JuSa00} devoted to the surfaces with linear normals \citep{Ju98}, which generally admit a non-proper PN parameterizations, see \cite{VrLa14a} for more explanations.

First we emphasize that our method requires that given points being interpolated are arranged in a rectangular grid; for further details about quadrilateral mesh generation and processing, including surface analysis and mesh quality, simplification, adaptive refinement, etc. we refer to survey paper \citep{BoLePiPuSiTaZo13} and references therein. Next, unlike the approaches presented in the papers cited in the previous paragraph we plan to interpolate a set of given points $\f p_{ij}$ with the associated normal vectors $\f n_{ij}$ by a rational parameterized PN surface in a direct way, i.e., without the necessary subsequent reparameterization. The advantage of such direct PN interpolation techniques is obvious -- no complicated trimming procedure in the parameter space is necessary. As far as we are aware, a similar method was discussed only in \citep{PePo96}, where a surface design scheme with triangular patches on parabolic Dupin cyclides was proposed. In addition, in \citep{Gravesen2007} the interpolation of triangular data using the support function is studied. The Gauss image is first constructed and then the support function interpolating the values and gradients at certain points (the given normals) is determined. Our approach interpolates the normals and the support function simultaneously in the isotropic space. This way we are able to produce local patches with global $G^1$ continuity.

In the beginning of this paper, we also very shortly address the related open problem of the Hermite interpolation with the {\em polynomial} PN surfaces. Rather than solve this problem, we show its complexity. Indeed, by simply considering the required number of free parameters we show that this problem is much harder than in the curve case. We use this consideration as a certain defense for using a dual technique, which leads to rational solutions. Even so we consider the Hermite interpolation with polynomial PN surfaces directly as a promising and challenging direction for our future research.

The remainder of this paper is organized as follows. Section~$2$ recalls some basic facts concerning surfaces with Pythagorean normal vector fields. We will also briefly sketch how to satisfy the PN condition in the polynomial case, i.e., how to find polynomial PN parameterizations. In Section 3 the representation of PN surfaces in the Blaschke and in the isotropic model is presented. We also discuss the usefulness of these representations for the solution of the interpolation problem. Section 4 is devoted to bicubic Coons patches in the isotropic model and their usage in the construction of smooth PN surfaces. The method is described,  discussed and presented on a particular example in Section~5. Finally, we conclude the paper in Section~6.

\section{Surfaces with Pythagorean normals}\label{sec prelim}

In this section we recall some fundamental facts about surfaces with rational offsets.

\begin{definition}%
Let $\av{X}$ be a real algebraic surface in $\R^3$, let $\av{X}^r$ denote the set of regular points of $\av{X}$,
and let us denote by $\f n_{\f p}\in \av{S}^2$ a unit normal vector at a point
$\f p\in\av{X}^r$. Then the \emph{$d$-offset} $\oo{X}$  of $\cal X$  is defined as the closure of the set
$\{\f p\pm d\f n_{\f p}\mid\,  \f p\in\av{X}^r\}$.
\end{definition}

If $\cal X$ is rational and  $\f x:\R^2\rightarrow \R^3$ is its parameterization, we may write down a parameterization of the offset explicitly in the form
\begin{equation}\label{eq param offset}
  \f x (u,v)\pm d \f n_{\f x}(u,v),
\end{equation}
where $\f n_{\f x}(u,v)$ is the unit normal vector field associated to the parameterization $\f x (u,v)$. It turns out that \eqref{eq param
offset} is rational if and only if $\f n_{\f x}(u,v)$ is. This is equivalent to the existence of a rational function $\sigma(u,v)$ such
that
\begin{equation}\label{PNcondition}
\|\f{x}_u\times\f{x}_v\|^2=\sigma^2,
\end{equation}
where $\f{x}_u$ and $\f{x}_v$ are partial derivatives with respect to $u$ and $v$, respectively.

\begin{definition}%
$C^1$ regular parametric surfaces fulfilling condition \eqref{PNcondition} are called {\em surfaces with Pythagorean normal vector fields} (or {\em PN surfaces}, in short) and condition \eqref{PNcondition} is referred to as \emph{PN condition} or \emph{PN property}.
\end{definition}

\medskip
PN surfaces were defined by \cite{Po95} as surface analogies to Pythagorean hodograph (PH) curves distinguished by the PH condition  $\|\f{x}'(t)\|^2=\sigma(t)^2$. These curves were introduced as planar polynomial objects. Later, the concept was generalized also to the rational PH curves, see \citep{Po95}. The interplay between the different approaches to polynomial and rational curves with Pythagorean hodographs was studied by \cite{FaPo96} and the former were established as a proper subset of the latter by presenting simple algebraic constraints.

Unfortunately more than 20 years from their introduction, the situation is still completely different for the PN surfaces. This is also reflected when solving the interpolation problems, in which the points and the normal vectors are prescribed as input data. The most natural (and expected) way of handling the PN surfaces would be probably similar to the one used for PH curves, see e.g. \cite{Fa08}. Let us show it on the polynomial case. All the polynomials satisfying the PH condition $x'(t)^2+y'(t)^2=\sigma(t)^2$ can be described explicitly using polynomial Pythagorean triples. The corresponding PH curve $\f x(t)=(x(t),y(t))$ is then  obtained simply by integration. In the surface case, however, we cannot reproduce this approach.

It is possible to describe explicitly all the polynomial Pythagorean normal fields $\f N(u,v)$ of degree $k$ having the polynomial length, i.e., $||\f N(u,v)||^2$ is a perfect square; cf. \citep{DiHoJu93}. To determine an associated PN parameterization of degree $\ell+1$ in a direct way, we have to find suitable polynomial vector fields
\begin{equation}
\begin{array}{c}
\displaystyle
\f P(u,v)=
\left(
\sum_{i+j\leq\ell}
\mbox{\hspace*{-0ex}}p_{1ij}{u^iv^j},
\sum_{i+j\leq\ell}
\mbox{\hspace*{-0ex}}p_{2ij}{u^iv^j},
\sum_{i+j\leq\ell}
\mbox{\hspace*{-0ex}}p_{3ij}{u^iv^j}
\right),\\[4ex]
\displaystyle
\f Q(u,v)=
\left(
\sum_{i+j\leq\ell}
\mbox{\hspace*{-0ex}}q_{1ij}{u^iv^j},
\sum_{i+j\leq\ell}
\mbox{\hspace*{-0ex}}q_{2ij}{u^iv^j},
\sum_{i+j\leq\ell}
\mbox{\hspace*{-0ex}}q_{3ij}{u^iv^j}
\right),
\end{array}
\end{equation}
which will play the role of $\f x_u$, $\f x_v$, respectively, that satisfy the following conditions
\begin{equation}\label{eq PN soustava}
\begin{array}{rcl}
\f P \cdot \f N & = & 0,\\
\f Q \cdot \f N & = & 0,\\[1ex]
\displaystyle \frac{\partial\f P}{\partial v} - \displaystyle \frac{\partial\f Q}{\partial u} & = & 0,
\end{array}
\end{equation}
where the third equation expresses the condition for the~integrability.  Since a polynomial of degree $n$ in two variables possesses $\binom{n+2}{2}$ coefficients, the problem is now transformed to solving a system of $2\binom{k+\ell+2}{2}+3\binom{\ell+1}{2}$ homogeneous linear equations with $6\binom{\ell+2}{2}$ unknowns $p_{1ij}, p_{2ij}, p_{3ij},q_{1ij},q_{2ij},q_{3ij}$. The corresponding PN parameterization is then obtain as
\begin{equation}
\f x(u,v)=\int\f {P}(u,v)\,\mathrm{d}u+\f {c}(v),
\mbox{ where}
\qquad
\f{c}(v)=\left[\int \f{Q}(u,v)\,\mathrm{d}v-\int \f{P}(u,v)\,\mathrm{d}u\right]_{u=0}.
\end{equation}

However, we must stress that not for every given polynomial Pythagorean normal field $\f N(u,v)$ there exists a corresponding polynomial surface $\mathbf x(u,v)$ for which  $\mathbf x_u \times \mathbf x_v=\f N(u,v)$. For this to hold we need $\ell=k/2$. Nevertheless, in this case the number of unknowns is less than the number of equations so one cannot expect a solution, in general. On the other hand for $\ell$ large enough, the system of equations \eqref{eq PN soustava} is solvable. In this case we obviously arrive at a PN parameterization such that $\mathbf x_u \times \mathbf x_v=f(u,v)\bf N(u,v)$, where $f(u,v)$ is some non-constant polynomial.

\section{PN surfaces in the isotropic model of the dual space}\label{PN_isotrop}

As the offsets have a considerably simplier description if we apply the dual approach, we recall in this section the representation of PN surfaces in the Blaschke and isotropic model of the dual space. Moreover, this concept is later used for formulating our Hermite interpolation algorithm.

For the sake of brevity, we exclude developable surfaces from our considerations and assume a non-degenerated Gauss image $\gamma(\av{X})$ of all studied surfaces $\av{X}$  in what follows. This means that the {\em duality} $\delta$ maps a surface $\av{X}$ to its {\em dual surface $\av{X}^*$}. Recall that a non-developable surface ${\cal X}: f(\f x)=0$  has the {\em dual representation}
\begin{equation}\label{dual}
 {\cal X}^*:\ F^*(\vek[n],h)=0,
\end{equation}
where $F^*$ is a homogeneous polynomial in $\f n=(n_1,n_2,n_3)$ and $h$. If $F^*(\f n,h)=0$ then the
set of all planes
\begin{equation}\label{tangents}
T_{\f n,h}: \f n\cdot \f x=h
\end{equation}
forms a~system of {\em tangent planes} of $\cal X$ with the normal vectors $\f n$ (i.e., $\av{X}^*$ is considered as the set of tangent planes of $\av{X}$). Furthermore, if we assume $\|\f n\|=1$ then the value of $h$ is the oriented distance of the tangent plane to the origin. Moreover, if the partial derivative $\partial F^*/\partial h$ does not vanish at $(\f n_0,h_0)\in {\cal X}^*$, then \eqref{dual} implicitly defines a function
\begin{equation}
\f n\mapsto h(\f n)
\end{equation}
in a certain neighborhood of $(\f n_0,h_0)$. The restriction of this function to the unit sphere $\av{S}^{2}$ is called the {\em support function} of the primal surface, see \cite{Gravesen2007,AiJuGVSc09,GrJuSi08,LaBaSi10,SiGrJu08} for more details. Let us stress out that that the dual representation \eqref{dual} does not require the normal vectors $\f n$ to be unit vectors. However, whenever we use the support function then its argument $\f n$ will be assumed to be a unit vector.

Conversely, from any smooth real function on (a subset of) $\av{S}^2$ we can reconstruct the corresponding primal surface by the parameterization $\f x_h:\,\av{S}^2\rightarrow\mathbb{R}^2$ \begin{equation}\label{envelope}
\f x_h(\f{n})=h(\f{n})\f{n}+\nabla_{\!\av{S}^2}h(\f n),
\end{equation}
where the vector $\nabla_{\!\av{S}^2}h$ is obtained by embedding the intrinsic gradient of $h$ with respect to $\av{S}^2$ into the space $\R^3$, see \cite{GrJuSi08} for more details. The vector-valued function $\f x_h$ gives a parameterization of the envelope of the set of tangent planes \eqref{tangents}. Hence, all surfaces with the associated rational support function are
rational. It is enough to substitute into \eqref{envelope} any rational parameterization of $\av{S}^2$, for instance ${\f{n}(u,v)=(2u/(1+u^2+v^2),2v/(1+u^2+v^2),(1-u^2-v^2)/(1+u^2+v^2))}$.

Furthermore, several important geometric operations correspond to suitable modifications of the support function, see \cite{SiGrJu08}. In particular the one-sided offset of a surface at the distance $d$ is obtained by adding the constant $d$ to the support function $h$. For using the support function for computing the convolutions (i.e., the general offsets) of two hypersurfaces see e.g. \cite{SiGrJu07}.

\bigskip
In \cite{PoPe98}, the rational surfaces with rational offsets were studied in the so-called Blaschke model. Consider in ${\mathbb{R}^4}$ the quadric  $\av{B}=\av{S}^2\times\R:\,  \|\f n\|^2-1=0$. This quadratic cylinder is called the {\em Blaschke cylinder}. It holds that parallel tangent planes are then represented as points lying on the same generator (a line parallel to the $x_4$-axis) of $\av{B}$. In what follows the map that
sends a point in $\av{X}$ to the tangent plane, i.e., a point in $\av{X}^*$, is called the {\em Blaschke mapping} and is denoted $\beta$.

\begin{proposition}
Any non-developable PN surface is the image of a rational surface on the Blaschke cylinder $\mathcal{B}$ via the mapping $\phi=\delta^{-1}\circ\beta$.
\end{proposition}

Next, consider the generator line $w$ of $\mathcal{B}$ containing the point $\f w= (0,0,1,0)$. Let $\av{I}$ be the hyperplane $x_3=0$ in ${\mathbb{R}^4}$, which is parallel to $w$. We use the new coordinate functions $y_1=x_1$, $y_2=x_2$,  $y_3=x_4$ and define the {\em isotropic mapping}
\begin{equation}
\iota:\quad \av{B}\setminus w\rightarrow \av{I},\, (x_1,x_2,x_3,x_4) \mapsto (y_1,y_2,y_3)=\frac{1}{1-x_3}(x_1,x_2,x_4).
\end{equation}
$\av{I}$ is called the {\em isotropic model} of the dual space, see Fig.~\ref{blaschke_isotrop}. Clearly, the tangent planes with the unit normal $(0,0,1)$ do not have an image point in $\av{I}$. Other parallel tangent planes are represented as points on the same line parallel to the $y_3$-axis; these lines are called the {\em isotropic lines}. By a direct computation one obtains
\begin{equation}
\iota^{-1}:\quad \av{I}\rightarrow \av{B}\setminus w,\, (y_1,y_2,y_3) \mapsto (x_1,x_2,x_3,x_4)=\frac{1}{1+y_1^2+y_2^2}(2y_1,2y_2,1-y_1^2-y_2^2,2y_3).
\end{equation}

\begin{figure}[t]
\begin{center}
  \psfrag{x3}{$x_3$}
  \psfrag{x4}{$x_4$}
  \psfrag{x1}{$x_1,x_2$}
  \psfrag{y3}{$y_3$}
  \psfrag{y1}{$y_1,y_2$}
  \psfrag{w}{$w$}
  \psfrag{W}{$\f w=(0,0,1,0)$}
  \psfrag{B}{$\av{B}$}
  \psfrag{I}{$\av{I}$}
  \psfrag{p}{\hspace*{-1ex}$\f p = (\f n,h)$}
  \psfrag{i}{$\iota(\f p)$}
  \psfrag{n}{$\f n$}
  \includegraphics[width=0.45\textwidth]{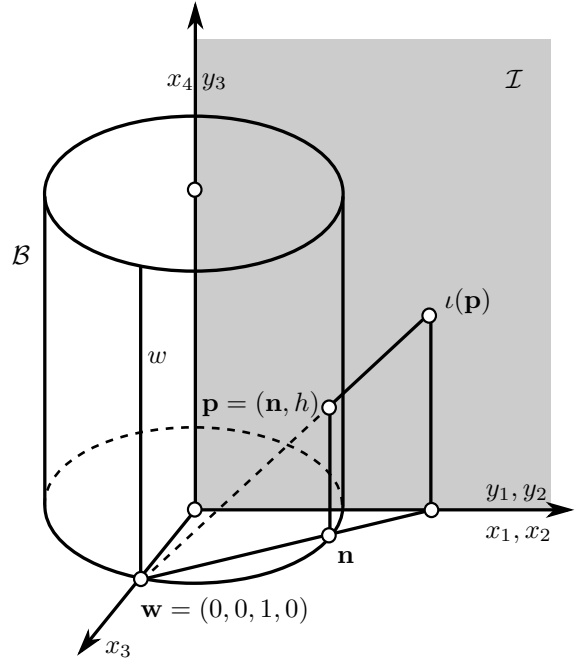}\hfill
\begin{minipage}{0.9\textwidth}
\caption{The Blaschke cylinder $\av{B}$ and the isotropic model $\av{I}$ of the dual space. \label{blaschke_isotrop}}
\end{minipage}
\end{center}
\end{figure}

All the above mentioned properties and mappings are summarized in the following proposition and diagram, which are essential for our method, see \eqref{eq diagram}.

\begin{equation}\label{eq diagram}
\begin{array}{c}
\scalebox{1.2}{
\xymatrix{
& \av{B}  \ar[r]_-{\iota}  \ar[d]^{\beta}  \ar[ld]_{\phi} & \ar@/^{-3pc}/[lld]_\xi \av{I} \ar[ld]^{\theta} \\
\av{X} \ar[r]^{\delta}  & \av{X}^* \vphantom{\bigr)}
}
}
\end{array}
\end{equation}

\begin{corollary}\label{PNisotrop}
Any non-developable PN surface  is the image of a rational surface in $\av{I}$ via the mapping $\xi=\phi\circ\iota^{-1}$.
\end{corollary}

\medskip
The effectiveness of the construction presented later is guaranteed by the following continuity result.

\begin{proposition}\label{G1primar}
Let $\mathbf y$ be a piecewise rational $C^1$ surface in $\av{I}$. If $\mathbf x=\xi(\mathbf y)$ is regular then it is a $G^1$ piecewise rational surface with Pythagorean normals.
\end{proposition}

\begin{proof}
By the regularity of $\mathbf x$ we mean that at every point there is a suitable tangent plane so that the projection of the surface to this plane is a homeomorphism on some neighborhood of this point. This essentially means that we exclude the sharp edges (ridges). Conditions for the regularity are discussed in more detail in the paragraph following this proof.
We have $\xi=\phi\circ\iota^{-1}$. The mapping $\iota$ (and its inverse) is a diffeomorphism which does not change the continuity. So the patch $\iota^{-1}(\mathbf y)$  on the Blaschke cylinder is clearly also $C^1$. The mapping  $\phi$ is given by formula \eqref{envelope}, which contains the first order differentiation. For this reason the surface $\mathbf x=\xi(\mathbf y)$ is only~$C^0$. However, in fact the continuity of $\iota^{-1}(\mathbf y)$ describes a continuous variation of a certain well defined plane.  It is shown in \citep{Gravesen2007} that if $\mathbf x$ is regular then the inversion of the projection to this plane is locally $C^1$ which shows the global $G^1$ continuity of $\mathbf x$.
\end{proof}

As it has been noticed earlier \citep{PePo96,Gravesen2007,SiGrJu08,Blazkova2014358} despite the fact that the variation of the planes in the previous proposition is continuous, the resulting surface may sometimes exhibit sharp edges (ridges). To understand this phenomena let us first investigate, for the sake of simplicity, two examples of planar curves, see Fig.~\ref{ridges}. In this case $h(\mathbf n)$ is univariate and the function $h+h''$ gives the oriented radius of curvature \citep{SiGrJu08}. If this expression vanishes the curve exhibits a cusp at which the curvature goes to infinity.

The first example is a part of a hypocycloid (Fig.~\ref{ridges}, left) where the support function $h(\mathbf n)$ is perfectly smooth but still a cusp occurs. The second example (Fig.~\ref{ridges}, right) shows two circle segments connected tangentially and producing a sharp jump in the signed curvature.

\begin{figure}[t]
\begin{center}
\begin{tabular}{cc}
  \includegraphics[height=0.18\textwidth]{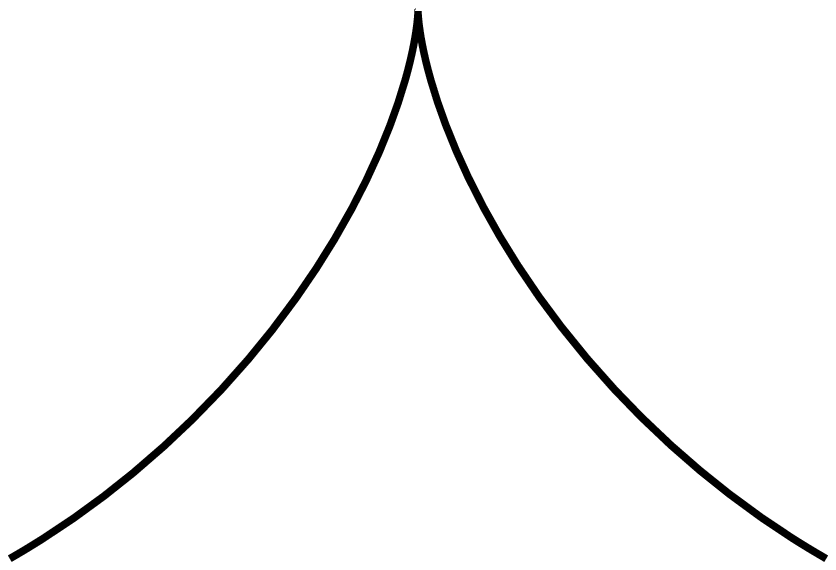}&
    \includegraphics[height=0.18\textwidth]{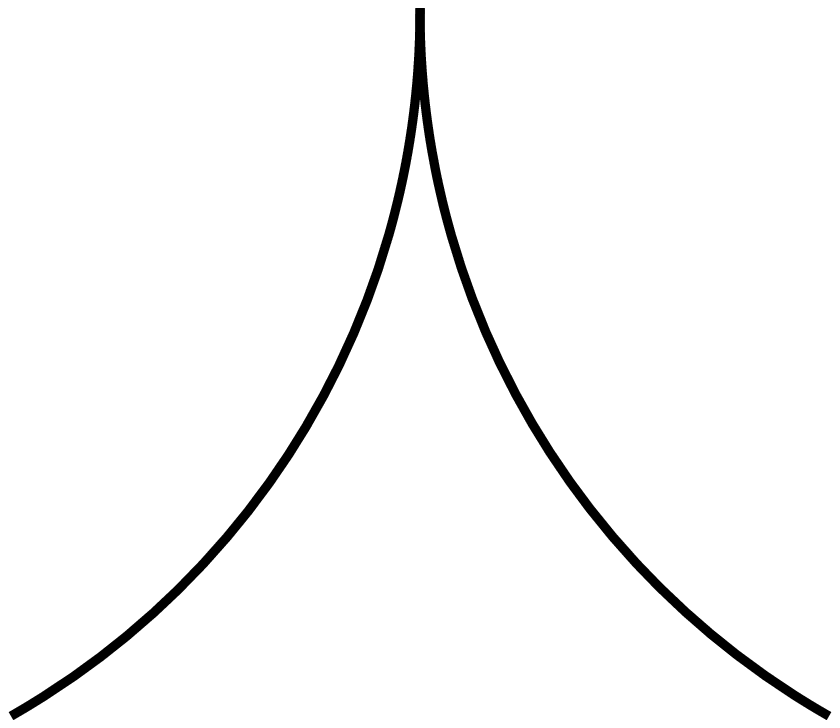}\\
     \includegraphics[width=0.33\textwidth]{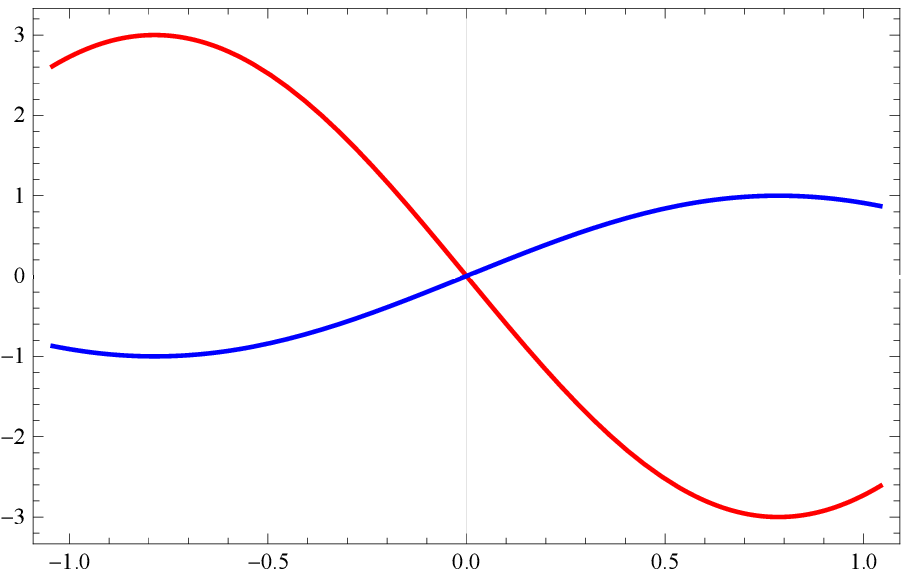}&
    \includegraphics[width=0.33\textwidth]{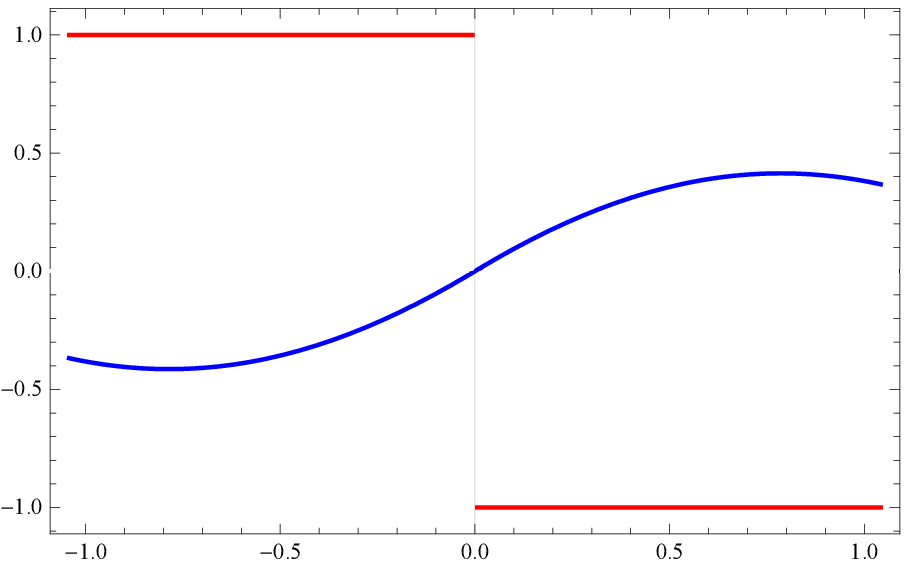}\\
\end{tabular}
\begin{minipage}{0.9\textwidth}
\caption{Curves with cusps and continuous tangent line. The support function (blue) is $C^\infty$ on the left and $C^1$ on the right. The value of $h+h''$ (radius of curvature) is displayed in red.\label{ridges}}
\end{minipage}
\end{center}
\end{figure}

\newcommand{\Hess}{{\mathrm{Hess}}}

A similar behavior can be described for surfaces. In this case the critical expression (corresponding to $h+h''$) is the matrix function $\Hess_{\mathcal{S}^2}h+hI$, where $\Hess_{\mathcal{S}^2}$ denotes the intrinsic Hessian with respect to the unit sphere $\mathcal{S}^2$ (the base of the Blaschke cylinder $\cal B$) and $I$ is the identity. In fact as shown in \citep{SiGrJu08}, it holds that $d \mathbf x_h=\Hess_{\mathcal{S}^2}h+hI$, so this quantity allows us to control the features of the resulting surface. Vanishing of the $\det\left(\Hess_{\mathcal{S}^2}h+hI\right)$ or a jump in the signs of one its eigenvalues indicates the occurrence of a sharp edge.  For practical modeling purposes let us remark, that a sharp edge typically occurs when the data from a surface with parabolic curves are interpolated. In other cases this phenomena will disappear under subdivision. Furthermore, as our method is based on the construction of Coons patches with boundaries being Fergusson cubics determined by suitably chosen tangent vectors at given points in the isotropic space (see Section~\ref{sec isotr}), it is theoretically also possible to avoid ridges by optimizing the lengths of the tangent vectors (which can serve as free modelling shape parameters, see Section~\ref{sec alg}) with a suitable objective function. One can for example use the function $\int_{\Omega} \det\left(\Hess_{\mathcal{S}^2}h+hI\right)^{-2}\mathrm{d}A_{\mathcal{S}^2}$, where $\mathrm{d}A_{\mathcal{S}^2}$ is the area element on the sphere and  ${\Omega}\subset \mathcal{S}^2$ is the Gauss image of the constructed surface. In fact $\det\left(\Hess_{\mathcal{S}^2}h+hI\right)^{-2}=K^2$ and when minimizing its integral we can avoid the ridges at which the Gauss curvature $K$ tends to infinity, see also \citep{Gravesen2007}.

\section{Coons patches in the isotropic model and PN patches in the primal space}\label{sec isotr}
We will use rectangular patches throughout this paper. In order to construct a piecewise PN interpolation surface $\mathbf x$ in the primal space, we will consider rational patches in the isotropic model.

\begin{figure}[t]
\begin{center}
  \psfrag{1}{$\f a_{00}$}
  \psfrag{2}{$\f c_0(u)$}
  \psfrag{3}{$\f a_{10}$}
  \psfrag{4}{$\f d_0(v)$}
  \psfrag{5}{$\f y(u,v)$}
  \psfrag{6}{$\f d_1(v)$}
  \psfrag{7}{$\f a_{01}$}
  \psfrag{8}{$\f c_1(u)$}
  \psfrag{9}{$\f a_{11}$}
  \includegraphics[width=0.45\textwidth]{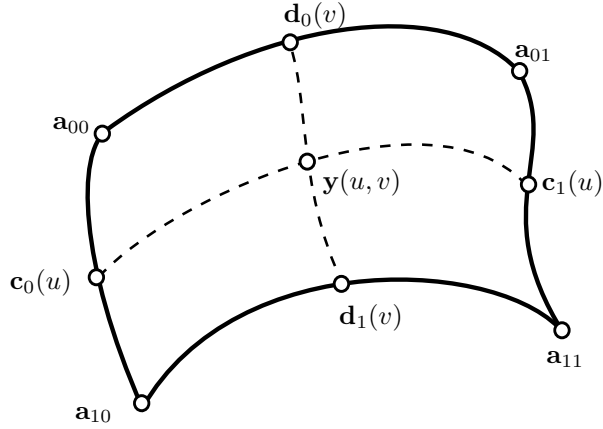}\hfill
\begin{minipage}{0.9\textwidth}
\caption{The Coons patch  $\f y(u,v)$ determined by \eqref{krivky}. \label{fig coons}}
\end{minipage}
\end{center}
\end{figure}

Suppose we are given four $C^1$ continuous boundary curves $c_0(u)$, $c_1(u)$, $d_0(v)$, $d_1(v)$ in the isotropic space $\av{I}$ which meet at the four corners
\begin{equation}\label{krivky}
c_0(0)=d_0(0)=\f a_{00},\quad
c_0(1)=d_1(0)=\f a_{10},\quad
c_1(0)=d_0(1)=\f a_{01},\quad
c_1(1)=d_1(1)=\f a_{11},
\end{equation}
see Fig.~\ref{fig coons}. Then we can apply the construction of the so called {\em bicubic Coons patch}, see e.g. \cite{Farin1988}. It is a parametric surface $\f y(u,v):\, [0,1]\times[0,1]\rightarrow \R^k$ ($k=3$ in our case) determined by the identity
\begin{equation}\label{Coons_bicubic}
\big(F_0(u),-1,F_1(u)\big)\cdot
\begin{pmatrix}
\f a_{00} & d_0(v) &  \f a_{01}\\
c_0(u) & \f y(u,v) &  \f c_1(u)\\
\f a_{10} & d_1(v) &  \f a_{11}
\end{pmatrix}
\cdot\big(F_0(v),-1,F_1(v)\big)^T
=0,
\end{equation}
where the blending functions $F_0,F_1$ are two of the basic cubic Hermite polynomials used in the construction of the Ferguson cubic, i.e., $F_0(t)=2t^3-3t^2+1$ and $F_1(t)=-2t^3+3t^2$.

We recall that the matrix in \eqref{Coons_bicubic} directly reflects the scheme in Fig.~\ref{fig coons}.
Formula \eqref{Coons_bicubic} ensures that the constructed patch interpolates all the given boundary curves  $c_0(u),c_1(u),d_0(v),d_1(v)$. Note that if only two tangent vectors at every point $\f a_{ij}$ instead of the whole boundary curves are given then one has to  first construct some boundary curves via interpolating these points and vectors by a suitable $C^1$ Hermite interpolation curves.

\medskip
By a direct computation \citep{Farin1988} it can be proved a fundamental property satisfied by the bicubically blended Coons patches

\begin{lemma}\label{C1Coons}
Two bicubic Coons patches sharing the same boundary curve and the same tangent vectors at the end points of the adjacent transversal boundary curves are connected with the $C^1$ continuity.
\end{lemma}

From this lemma follows one of the nicest application properties of the bicubic Coons construction. Specifically, given a network of curves, the global interpolating surface that one gets using the bicubic Coons construction is globally a $C^1$ surface. Combined with  Proposition~\ref{G1primar} we obtain the fundamental theoretical result justifying our Hermite PN construction.

\begin{proposition}\label{G1coons}
Let $\f{y}$ be a globally $C^1$ continuous network of piecewise rational Coons patches in the space $\av{I}$. Then $\f{x}=\xi(\f{y})$ is a piecewise $G^1$ surface with Pythagorean normals.
\end{proposition}

The observations and results above allow us to design a simple construction algorithm which is essentially local. More precisely, for a given network of position data (points) and first order data (normals) we will construct a family of PN patches yielding a piecewise surface which is globally $G^1$ continuous. Specifically, a modification of some of these data will modify only the adjacent patches.

\bigskip
Suppose we are given a network of the points $\f p_{i,j}$ with the associated unit normal vectors $\f n_{i,j}$ in the primal space, where $i\in \{0,1,\ldots, m \}$ and $j\in \{0,1,\ldots, n \}$. Our goal is to construct a set of rational PN patches $\f x_{i,j}(u,v)$ for  $i\in \{1,\ldots, m \}$, $j\in \{1,\ldots, n \}$. Each patch will be defined on the interval $[0,1]\times [0,1]$ and will interpolate the corner points $\f p_{i-1,j-1}$, $\f p_{i,j-1}$, $\f p_{i-1,j}$, $\f p_{i,j}$ together with the corresponding normals. Moreover the union of these patches $\f x=\bigcup_{i,j}\f x_{i,j}$ is required to be globally $G^1$ continuous.

Based on the theoretical results from the previous sections, we will construct the patches  $\f x_{i,j}(u,v)$ as the images of the rational patches $\f y_{i,j}(u,v)$ in the isotropic space $\av{I}$, i.e.,
\begin{equation}
\f x_{i,j}(u,v)=\xi(\f y_{i,j}(u,v)).
\end{equation}
First, for each point we evaluate the support function $h_{i,j}= \f p_{i,j} \cdot \f n_{i,j}$, cf. \eqref{tangents}. Next we obtain the corresponding network of points in the isotropic space $\av{I}$ as
\begin{equation}\label{transfbody}
\f a_{i,j}=\iota(\f n_{i,j},h_{i,j}).
\end{equation}

\medskip
In order to apply the bicubic Coons patch construction, we need to construct boundary curves between the points $\f a_{i,j}$. From the identity
\begin{equation}
\f n(u,v) \cdot \f x (u,v) - h (u,v) =0,
\end{equation}
it follows
\begin{equation}\label{der_support}
\begin{array}{rcl}
\f (n_u,h_u)\cdot (\f p,-1) &=& 0,\\
\f (n_v,h_v) \cdot (\f p,-1) &=& 0.
\end{array}
\end{equation}
So, let us observe that any curve $\f c(t)$ lying on the piecewise surface $\f{y}$ such that $\f c(t_0)=\mathbf a_{i,j}$ must satisfy
\begin{equation}
[J (\iota^{-1}) \f c' (t_0)]\cdot (\f p_{i,j},-1)=0,
\end{equation}
where $J (\iota^{-1})$ denotes the Jaccobi matrix of the mapping $\iota^{-1}$. It means that the patch possessing $\f a_{i,j}$ as its corner point (in the isotropic space) must be tangent to the 2-plane $\tau_{i,j}$ given as
\begin{equation}
\tau_{i,j}=\{ \f v:\, [J (\iota^{-1}) \f v]\cdot (\f p_{i,j},-1)=0\}
\end{equation}
at this point.

Let us stress out that the  original PN interpolation problem (prescribed points and normal vectors, i.e., tangent planes) in the primal space was difficult to solve. Using the methods presented above we have transformed it to the same kind of the interpolation problem (prescribed points and normal vectors, i.e., tangent planes), now in the isotropic space~$\av{I}$. However, after the transformation we do not have to care about the PN property --  this property is now obtained for free.

\begin{remark}\rm
One limitation of the presented method should be noted. As the north pole $\f w$, see Fig.~\ref{blaschke_isotrop}, is the center of the stereographic projection, the points on the unit sphere $\mathcal{S}^2$ (the unit normals $\f n_{i,j}$) must be suitably distributed. In other words, the Gauss image of the interpolating surface cannot contain  $\f w$. This means that in some cases a preliminary coordinate transformation is needed.

Let us also remark that one can alternatively interpolate the Gauss image (given data $\f n_{i,j}$) and the support function (data $h_{ij}$ computed from given data $\f n_{i,j}$ and $\f p_{i,j}$) separately, cf. \citep{Gravesen2007}. Firstly, one interpolates data  $\f n_{i,j}$ by a piecewise rational $C^1$ surface on $\mathcal{S}^2$, see e.g. \citep{AlNeSchu96}. Then using \eqref{der_support} we arrive at the values of the partial derivatives  $h_u, h_v$ at the points $\f p_{i,j}$ and computing e.g. one-dimensional Coons patches we arrive at the piecewise $C^1$ function $h(u,v)$. The sought PN parameterization is obtained just by switching from the dual to the primary space. For the sake of lucidity we prefer to apply the isotropic model as this approach is more illustrative and needs less steps.
\end{remark}

\section{PN patches interpolating given data}\label{sec alg}

To start the Coons construction in $\av{I}$, we must first construct curves  $\f c_{i,j}(u)$ connecting the points $\f a_{i,j}$ and  $\f a_{i+1,j}$ and curves $\f d_{i,j}(v)$ connecting the points $\f a_{i,j}$ and  $\f a_{i,j+1}$, simultaneously satisfying the condition that they are tangent to the planes $\tau_{i,j}$ at each of the two boundary points.
Clearly, any arbitrary curve fulfilling these constraints may be considered as one of the input boundary curves for scheme \eqref{Coons_bicubic}. For the sake of simplicity we can, for instance, take the Ferguson cubics interpolating with $C^1$ continuity the given points and some suitably chosen associated boundary vectors. Other possible polynomial curves of low parameterization degree, which can be easily used, might be e.g. parabolic biarcs.

\smallskip
We have considered the following boundary vectors at the points from the network in $\av{I}$, which represent a natural choice for the tangent vectors of the boundary curves:
\begin{itemize}
\item
for an inner point $\f a_{i,j}$ (see Fig.~\ref{points_in_net}, green) we have taken the projections of the difference vectors $\f a_{i+1,j}-\f a_{i-1,j}$ and  $\f a_{i,j+1}-\f a_{i,j-1}$ into the tangent plane $\tau_{i,j}$;
\item
for a non-corner point $\f a_{i,0}$, or $\f a_{i,n}$ on the $u$-boundary (see Fig.~\ref{points_in_net}, blue) we have taken the projections of the difference vectors $\f a_{i+1,0}-\f a_{i-1,0}$ and $2(\f a_{i,1}-\f a_{i,0})$, or $\f a_{i+1,n}-\f a_{i-1,n}$ and $2(\f a_{i,n}-\f a_{i,n-1})$, respectively,  into the tangent plane $\tau_{i,0}$, or $\tau_{i,n}$, respectively;
\item in a similar way, for a non-corner point $\f a_{0,j}$, or $\f a_{n,j}$ on the $v$-boundary (see Fig.~\ref{points_in_net}, blue) we have taken the projections of the difference vectors $2(\f a_{1,j}-\f a_{0,j})$ and $\f a_{0,j+1}-\f a_{0,j-1}$, or $2(\f a_{n,j}-\f a_{n-1,j})$ and $\f a_{n,j+1}-\f a_{n,j-1}$, respectively,  into the tangent plane $\tau_{0,j}$, or $\tau_{n,j}$, respectively;
\item
for the corner point $\f a_{0,0}$ (see Fig.~\ref{points_in_net}, red) we have taken the projections of the difference vectors $2(\f a_{1,0}-\f a_{0,0})$ and  $2(\f a_{0,1}-\f a_{0,0})$ into the tangent plane $\tau_{0,0}$,
for the corner point $\f a_{n,0}$ we have taken the projections of the difference vectors $2(\f a_{n,0}-\f a_{n-1,0})$ and  $2(\f a_{n,1}-\f a_{n,0})$ into the tangent plane $\tau_{n,0}$,
for the corner point $\f a_{0,n}$ we have taken the projections of the difference vectors $2(\f a_{1,n}-\f a_{0,n})$ and  $2(\f a_{0,n}-\f a_{0,n-1})$ into the tangent plane $\tau_{0,n}$, and
for the corner point $\f a_{n,n}$ we have taken the projections of the difference vectors $2(\f a_{n,n}-\f a_{n-1,n})$ and  $2(\f a_{n,n}-\f a_{n,n-1})$ into the tangent plane $\tau_{n,n}$.
\end{itemize}

Of course, the lengths of the chosen vectors can be easily modified and serve as possible modelling shape parameters. This is useful, for instance, when we want to
avoid ridges by optimizing these lengths with respect to a suitable objective function, cf. the final paragraph in Section~\ref{PN_isotrop}. Subsequently, we construct the rational patches $\f y_{i,j}$ using formula \eqref{Coons_bicubic} and applying $\xi$ we obtain the patches $\f x_{i,j}$ and thus the sought piecewise smooth PN surface $\f{x}$.

\begin{figure}[t]
\begin{center}
  \includegraphics[width=0.4\textwidth]{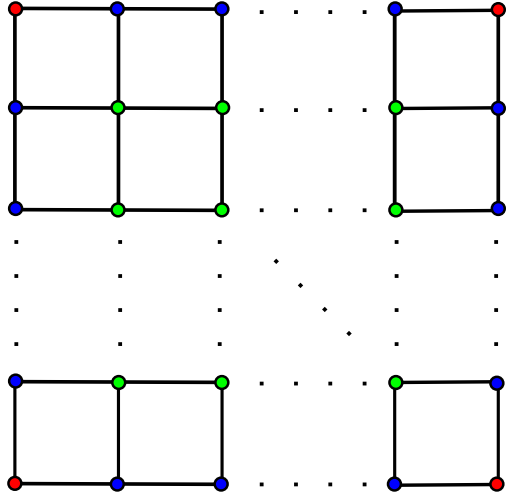}\hfill
\begin{minipage}{0.9\textwidth}
\caption{A net of points in $\av{I}$ -- corner points (red), non-corner boundary points (blue) and inner points (green). \label{points_in_net}}
\end{minipage}
\end{center}
\end{figure}

\medskip
In what follows we will show the functionality of the designed algorithm on a particular example. We will demonstrate the whole technique on one macro-element consisting of nine ordered points with the associated normals, i.e., a smooth surface consisting of four PN patches is constructed. For a bigger network the process will be the same, as the designed method is strictly local.

\begin{example}\rm
Let be given a network of the points $\f p_{i,j}$

\begin{equation}
(\f p_{i,j})=\left(
\begin{array}{ccc}
 \displaystyle (0,0,0) & \displaystyle \left(0,-\frac{11}{72},-\frac{1}{12}\right) & \displaystyle \left(0,-\frac{2}{9},-\frac{1}{3}\right) \\[2ex]
 \displaystyle \left(\frac{11}{72},0,\frac{1}{12}\right) & \displaystyle \left(\frac{7}{36},-\frac{7}{36},0\right) & \displaystyle \left(\frac{23}{72},-\frac{11}{36},-\frac{1}{4}\right) \\[2ex]
 \displaystyle \left(\frac{2}{9},0,\frac{1}{3}\right) & \displaystyle \left(\frac{11}{36},-\frac{23}{72},\frac{1}{4}\right) & \displaystyle \left(\frac{5}{9},-\frac{5}{9},0\right)
\end{array}
\right)
\end{equation}
with the associated (non-unit) normal vectors $\f n_{i,j}$
\begin{equation}
(\f n_{i,j})=\left(
\begin{array}{ccc}
 (0,0,-1) & (0,4,-3) & (0,1,0) \\[1ex]
 (4,0,-3) & (2,2,-1) & (4,8,1) \\[1ex]
 (1,0,0) & (8,4,1) & (2,2,1)
\end{array}
\right),
\end{equation}
where $i,j=0,1,2$, see Fig.~\ref{zadani}.

\begin{figure}[t]
\begin{center}
   \includegraphics[width=0.6\textwidth]{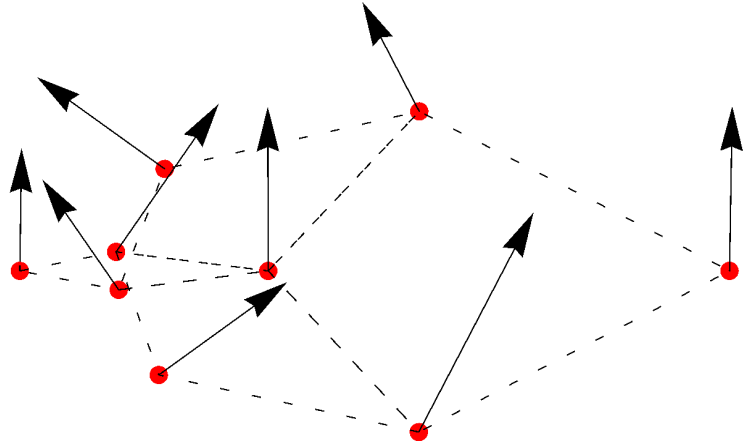}
\begin{minipage}{0.9\textwidth}
\caption{A network of given  points $\f p_{i,j}$ with the associated normal directions $\f n_{i,j}$. \label{zadani}}
\end{minipage}
\end{center}
\end{figure}

\smallskip
Using \eqref{transfbody} we find the nine points  $\f a_{i,j}$ (4 corner points, 4 non-corner boundary points, 1 inner point) in $\av{I}$, see Fig.~\ref{isotropic}, with the associated tangent vectors of the boundary curves obtained by the approach from the beginning of this section. Next, we construct 12 Fergusson cubics, see Fig.~\ref{isotropic}, as the input boundary curves for the bicubic Coons construction.

\begin{figure}[H]
\begin{center}
\includegraphics[width=0.7\textwidth]{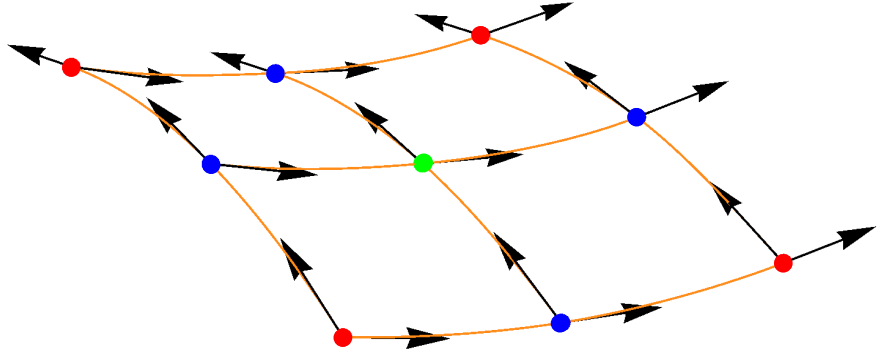}
\begin{minipage}{0.9\textwidth}
\caption{4 corner points (red), 4 non-corner boundary points (blue), 1 inner point(green) in $\av{I}$ with the associated tangent vectors of the boundary curves spanning the tangent planes $\tau_{ij}$, and the constructed Fergusson cubics (orange). \label{isotropic}}
\end{minipage}
\end{center}
\end{figure}

After computing the four bicubic Coons patches and applying the mapping  $\xi$ on each of them we obtain a smooth piecewise PN surface (given by PN parameterizations of each patch) interpolating given Hermite data, see Fig.~\ref{reseni}. Finally, computations show that $\det\left(\mathrm{Hess}_{\mathcal{S}^2}h+hI\right)\neq 0$ at all points so no sharp edges occur for given data, cf.~Section~3.

\begin{figure}[ht]
\begin{center}
\includegraphics[width=0.7\textwidth]{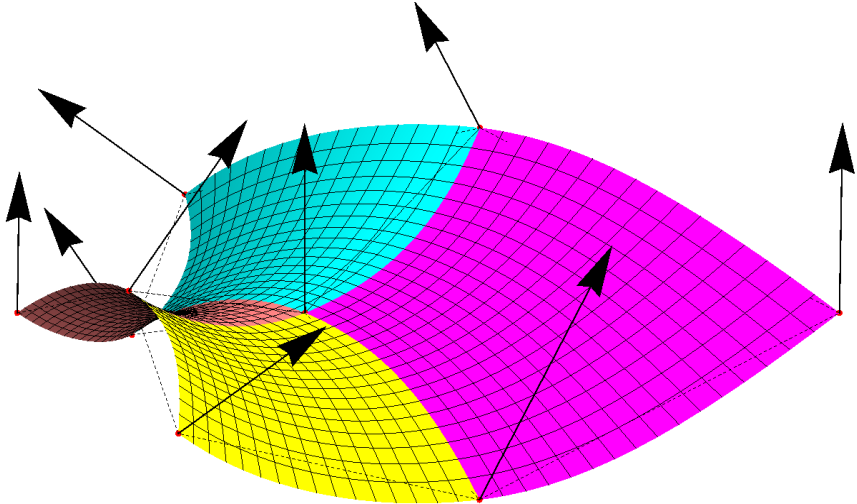}
\begin{minipage}{0.9\textwidth}
\caption{A smooth piecewise interpolation surface consisting of four PN patches. \label{reseni}}
\end{minipage}
\end{center}
\end{figure}

\end{example}

\begin{remark}\rm
One of the advantages of the designed method (based on exploiting the Coons patches) is the possibility to use the length of the tangent vectors in the isotropic space as free construction parameters (as already mentioned at~the end of Section~\ref{PN_isotrop}). Figure \ref{Sridges} shows how a suitable choice of these vector can improve the resulting patch and help to avoid the ridges.

\begin{figure}[tbh]
\begin{center}
   \includegraphics[width=0.45\textwidth]{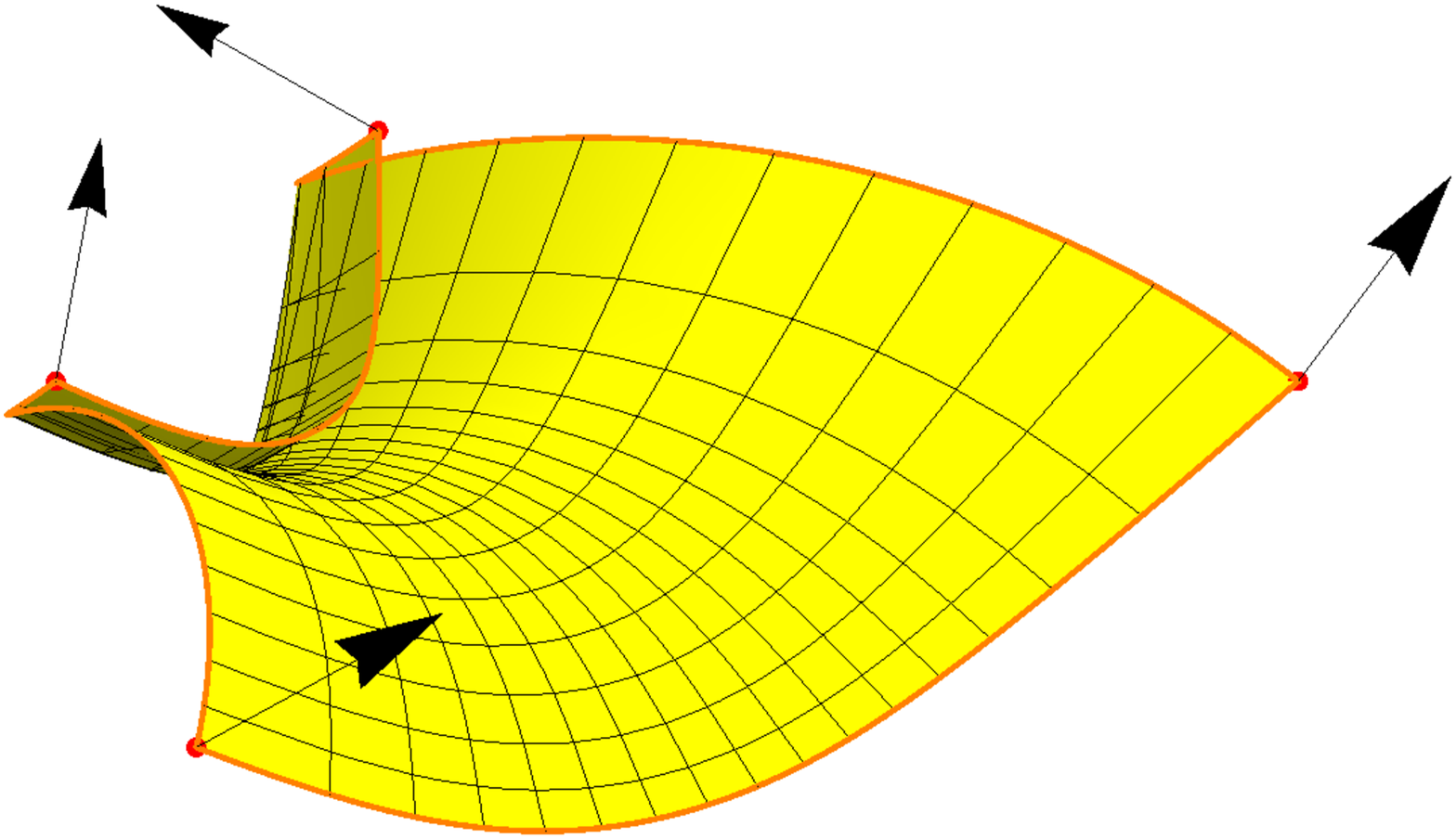}\qquad    \includegraphics[width=0.45\textwidth]{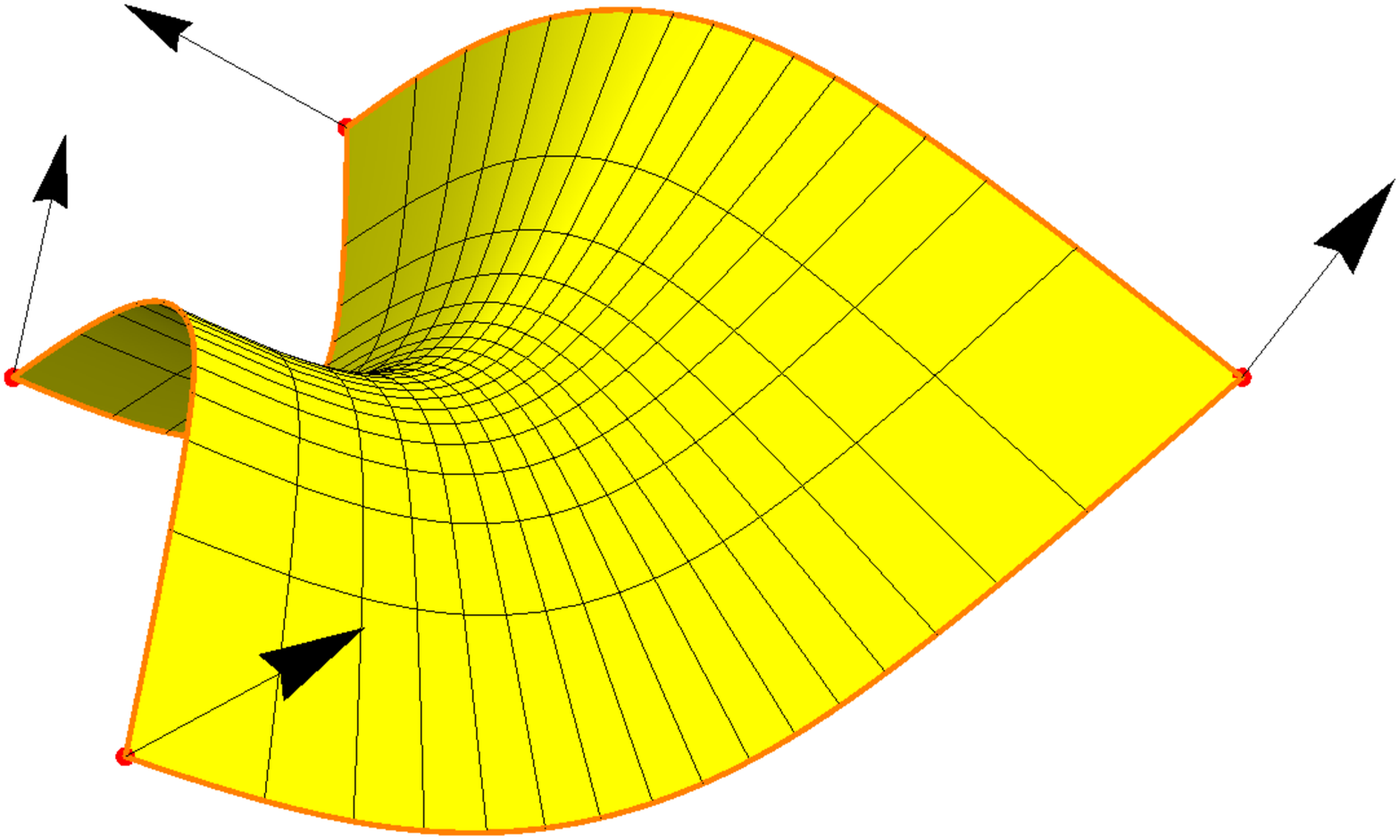}
\begin{minipage}{0.9\textwidth}
\caption{PN patches interpolating the same boundary data. A suitable choice of the tangent vectors leads to a smooth patch (right), while some choices may produce ridges (left).\label{Sridges}}
\end{minipage}
\end{center}
\end{figure}

\end{remark}

\begin{remark}
\rm
A natural question is why not to use the Coons (or some other boundary-curves) construction already in the primal space. A possible way could be for instance to prescribe some boundary curves satisfying given data, construct a patch given by this boundary and then to modify suitably the obtained patch (simultaneously preserving the conditions at the boundary) to gain a new patch which is PN. However, this construction assumes a necessary requirement that the prescribed curves must be PSN, i.e., curves on the surfaces along which the surface admits Pythagorean normals, cf. \cite{VrLa14a}. Using the dual approach and the isotropic model for this is considerably simpler.
\end{remark}

\section{Conclusion}\label{Concl}

The main goal of this paper was to present a simple functional algorithm for computing piecewise Hermite interpolation surfaces with rational offsets. The obtained PN surface interpolates a set of given points with associated normal directions. The isotropic model of the dual space was used for formulating the algorithm. This setup enables us to apply the standard bicubic Coons construction in the dual space for obtaining the interpolation PN surface in the primal space. The presented method is completely local and yields a surface with $G^1$~continuity. Moreover the method solves the PN interpolation problem directly, i.e., without the need for any subsequent reparameterization, which must be always followed by trimming of the parameter domain. Together with its simplicity, this is a main advantage of the designed technique. It can be used by designers anytime when surfaces with rational offsets are required for modelling purposes.

\section*{Acknowledgments}

The authors Miroslav L\'{a}vi\v{c}ka and Jan Vr\v{s}ek were supported by the project LO1506 of the Czech Ministry of Education, Youth and Sports.
We thank to all referees for their valuable comments, which helped us to improve the paper.

\end{document}